%% file: main-eptcs-proc.tex
\documentclass{eptcs}

\usepackage{amsmath}
\usepackage{amsthm}
\usepackage{amssymb}
\usepackage[mathscr]{eucal}

\usepackage{paralist}  % for compactitemize: asparaitem

\usepackage{rotating}
\usepackage{color}
\usepackage{graphicx}
\usepackage{subfig}

\usepackage{cite}
\usepackage{xspace}
\usepackage{url}

\input{shortcuts}

\renewcommand{\fix}[1]{}

\renewcommand{\bigOh}{\mathscr{O}}
\renewcommand{\propset}{\mathcal{P}}
\renewcommand{\varset}{\mathcal{V}}
\renewcommand{\Diamond}{\operatorname{\text{\raisebox{-.1em}{\begin{turn}{45}\scalebox{.9}{$\Box$}\end{turn}}}}}

\newtheorem{theorem}{Theorem}[section]
\newtheorem{lemma}[theorem]{Lemma}

\newtheorem{corollary}[theorem]{Corollary}
\newtheorem{definition}[theorem]{Definition}

\title{The $\mu$-Calculus Alternation Hierarchy Collapses \\ over
  Structures with Restricted Connectivity}

\def\titlerunning{The $\mu$-Calculus over Structures with Restricted
  Connectivity}

\author{Julian Gutierrez
  \institute{University of Cambridge, United Kingdom}
  \and 
  Felix Klaedtke
  \institute{ETH Zurich, Switzerland} 
  \and 
  Martin Lange
  \institute{University of Kassel, Germany}}

\def\authorrunning{J.\ Gutierrez, F.\ Klaedtke \& M.\ Lange}

\def\event{3rd International Symposium on 
  Games, Automata, Logics, and Formal Verification}

\begin{document}

\maketitle

\begin{abstract}
  It is known that the alternation hierarchy of least and greatest
  fixpoint operators in the \mbox{$\mu$-calculus} is strict.  However,
  the strictness of the alternation hierarchy does not necessarily
  carry over when considering restricted classes of structures.  A
  prominent instance is the class of infinite words over which the
  alternation-free fragment is already as expressive as the full
  \mbox{$\mu$-calculus.}  Our current understanding of when and why
  the \mbox{$\mu$-calculus} alternation hierarchy is not strict is
  limited.  This paper makes progress in answering these questions by
  showing that the alternation hierarchy of the \mbox{$\mu$-calculus}
  collapses to the alternation-free fragment over some classes of
  structures, including infinite nested words and finite graphs with
  feedback vertex sets of a bounded size.  Common to these
  classes is that the connectivity between the components in a
  structure from such a class is restricted in the sense that the
  removal of certain vertices from the structure's graph decomposes it
  into graphs in which all paths are of finite length.  Our collapse
  results are obtained in an automata-theoretic setting.  They
  subsume, generalize, and strengthen several prior results on the
  expressivity of the \mbox{$\mu$-calculus} over restricted classes of
  structures.
\end{abstract}

\setlength{\abovedisplayskip}{8pt plus 3pt minus 5pt}
\setlength{\belowdisplayskip}{8pt plus 3pt minus 5pt}

\input{intro-proc}
\input{prelim-proc}

\input{trans-proc}
\input{constr-proc}
\input{aflmuclasses-proc}
\input{concl-proc}

\bibliographystyle{eptcs}
\bibliography{refs}

\end{document}

%% file: shortcuts.tex
\newcommand{\ignore}[1]{}

\newcommand{\fix}[1]{\textcolor{red}{\footnote{\textcolor{red}{{\bf FixMe!!} #1}}}}

\newcommand{\dbrackl}{[\hspace{-0.13em}[}
\newcommand{\dbrackr}{]\hspace{-0.13em}]}

\newcommand{\set}[1]{\{#1\}}

\newcommand{\setx}[2]{\{#1\mathbin{\,\mid\,}#2\}}
\newcommand{\Setx}[2]{\big\{#1\mathbin{\,\big\vert\,}#2\big\}}

\newcommand{\Nat}{\mathbb{N}}

\newcommand{\Bool}{\mathbf{B}}

\renewcommand{\inf}{\operatorname{inf}}
\newcommand{\size}[1]{|\!|#1|\!|}
\newcommand{\ind}[1]{\operatorname{ind}(#1)}
\newcommand{\release}{\operatorname{release}}

\newcommand{\ad}[1]{\operatorname{ad}(#1)}

\newcommand{\BDAG}{\mathbf{BDAG}}

\renewcommand{\phi}{\varphi}
\renewcommand{\rho}{\varrho}
\renewcommand{\epsilon}{\varepsilon}

\newcommand{\bigOh}{\mathcal{O}}

\def\mup#1#2{\relax\ifmmode{}^{#1}{\protect#2}\else{}^{#1}#2\fi}
\def\nup#1#2{\relax\ifmmode{}_{#1}{\protect#2}\else{}_{#1}#2\fi}

\makeatletter

\def\create@abbrev#1#2#3{
  \def\c@a@def##1{
      \if ##1.
        \relax
      \else
        \@ifdefinable{\@nameuse{#1##1}}{\@namedef{#1##1}{#2##1}}
        \expandafter\c@a@def
      \fi
    }
  \c@a@def #3.
}

\create@abbrev{bf}{\mathbf}{ABCDEFGHIJKLMNOPQRSTUVWXYZ}
\create@abbrev{cal}{\mathcal}{ABCDEFGHIJKLMNOPQRSTUVWXYZ}
\create@abbrev{scr}{\mathscr}{ABCDEFGHIJKLMNOPQRSTUVWXYZ}
\create@abbrev{frak}{\mathfrak}{ABCDEFGHIJKLMNOPQRSTUVWXYZ}
\create@abbrev{bb}{\mathbb}{ABCDEFGHIJKLMNOPQRSTUVWXYZ}
\create@abbrev{tt}{\mathtt}{ABCDEFGHIJKLMNOPQRSTUVWXYZabcdefghijklmnopqrstuvwxyz}
\create@abbrev{rm}{\mathrm}{ABCDEFGHIJKLMNOPQRSTUVWXYZabcdefghijklmnopqrstuvwxyz}
\create@abbrev{sf}{\mathsf}{ABCDEFGHIJKLMNOPQRSTUVWXYZabcdefghijklmnopqrstuvwxyz}

\create@abbrev{ol}{\overline}{ABCDEFGHIJKLMNOPQRSTUVWXYZ}

\create@abbrev{aut}{\mathscr}{ABCDEFGHIJKLMNOPQRSTUVWXYZ}

\makeatother

\newcommand{\sem}[1]{\mathopen{\dbrackl}#1\mathclose{\dbrackr}}

\newcommand{\propset}{\mathscr{P}}
\newcommand{\varset}{\mathscr{V}}
\newcommand{\TRUE}{\ensuremath{\mathsf{tt}}}
\newcommand{\FALSE}{\ensuremath{\mathsf{ff}}}

\newcommand{\DIAMOND}[1]{\operatorname{{\mathopen{\langle}}\mathit{#1}{\mathclose{\rangle}}}}
\newcommand{\BOX}[1]{\operatorname{\mathopen{[}\mathit{#1}\mathclose{]}}}

\newcommand{\mucalc}{\ensuremath{\mathcal{L}_\mu}\xspace}

%% file: intro-proc.tex
\section{Introduction}

The $\mu$-calculus~\cite{Kozen1983:mucalculus}, hereafter \mucalc,
extends modal logic with least and greatest fixpoint operators, which
act as monadic second-order (MSO) quantifiers within the logic.
The possibility to arbitrarily mix and nest fixpoint operators makes
\mucalc an expressive logic, which subsumes many dynamic, temporal,
and description logics such as PDL and CTL\textsuperscript{*}.  In
fact, \mucalc is essentially the most expressive logic of that kind as
it can express, up to bisimulation equivalence, all MSO-definable
properties~\cite{Janin_Walukiewicz1996:mso_mucalulus}.

An important question about the expressivity of \mucalc is whether
more alternation---the nesting of mutually dependent least and greatest
fixpoint operators in formulas---gives more expressive power.
Bradfield~\cite{Bradfield1998:mucalculus_strict} proved that indeed
this is in general the case, i.e., there is a hierarchy of properties
that require unbounded alternation of least and greatest fixpoint
operators.  Lenzi~\cite{ICALP::Lenzi1996} independently showed a
similar strictness result---for a fragment of \mucalc---with respect
to an alternation hierarchy different from the one we consider in this
paper.  In both cases, their strictness results apply to the class of
finite directed graphs and therefore to all bigger classes of
structures.  However, the strictness of the alternation hierarchy need
not necessarily carry over when considering classes of structures that
are either incomparable to or smaller than the class of finite
directed graphs.  Trivial examples over which the alternation
hierarchy is non-strict are classes that only consist of a single
graph.  Here, each formula is equivalent to either \emph{true} or
\emph{false}, depending on whether the graph satisfies the formula or
not.

Overall, little is known about the expressivity of \mucalc over
restricted classes of structures.
Since \mucalc is bisimulation-invariant and every finite graph, either
directed or undirected, is bisimilar to a possibly infinite tree, the
strictness of the hierarchy also holds for the class of trees.  In
fact, as shown by Arnold~\cite{RAIRO:arnold99} and
Bradfield~\cite{Bradfield99-binarytree}, the hierarchy is strict even
on the class of binary infinite trees.
Alberucci and Facchini~\cite{Alberucci_Facchini2009:transitive_reflexive}
also strengthened the initial strictness result by showing that the hierarchy
remains strict over the class of reflexive finite directed graphs.

On the opposite side, there are a few classes of structures over which
it is known that the alternation hierarchy is not strict.
For instance, the hierarchy collapses to its alternation-free fragment
over the class of finite directed acyclic
graphs~\cite{Mateescu:2002:LMC}.  That is, for every \mucalc formula
$\phi$, there is an alternation-free \mucalc formula $\psi$, i.e., one
in which least and greatest fixpoint operators do not mutually depend
on each other, such that $\phi$ and $\psi$ are satisfied by exactly
the same set of finite acyclic graphs.  This collapse result is not
too surprising since 
the denotation of the least and greatest fixpoint operators of \mucalc 
differs only when considering models which contain infinite paths---and 
finite directed acyclic graphs only contain finite paths. 
Thus, in
this case, every greatest fixpoint operator can be replaced by a least
one, resulting in an alternation-free formula.
It is also known, when restricting \mucalc to infinite words,
that the \mucalc's alternation hierarchy collapses to its alternation-free
fragment~\cite{Kaivola1995:multl}. Moreover, over infinite nested
words, as shown by Arenas
et~al.~\cite{Arenas_Barcelo_Libkin2011:nestedwords}, the alternation
hierarchy collapses to the fragment with at most one alternation
between least and greatest fixpoint operators.
Finally, it is known that \mucalc's alternation hierarchy collapses
over the class of transitive finite directed
graphs~\cite{Alberucci_Facchini2009:transitive_reflexive,Dawar_Otto2009:modal_characterizations,DAgostino_Lenzi2010:transitive}.
If the graphs are transitive and undirected, then the hierarchy even
collapses to the modal
fragment~\cite{Alberucci_Facchini2009:transitive_reflexive,Dawar_Otto2009:modal_characterizations}.

This paper provides further classes of structures over which the
alternation hierarchy of \mucalc collapses to its alternation-free
fragment.  In fact, our collapse results subsume, generalize, and
strengthen some of the collapse results mentioned above.
In particular, we show that the alternation hierarchy collapses over
classes of finite directed graphs with feedback vertex sets of
a bounded size.  Recall that removing the vertices in a feedback vertex
set decomposes the graph into finite directed acyclic graphs and thus
the removal of these vertices eliminates the infinite behavior in the
original graph.  Finite directed acyclic graphs have the empty set as
feedback vertex set.
We also show that, as for infinite words, all \mucalc properties of
infinite nested words can already be expressed within the
alternation-free fragment.  Our collapse results are obtained in a
uniform way by looking at bounded classes of so-called bottlenecked
directed acyclic graphs.  The vertices of such a kind of graphs are grouped
into layers and the infinite paths must visit infinitely often
vertices in certain layers, which are bounded in their
size. Intuitively speaking, these bounded layers are the bottlenecks
and the removal of these vertices disconnects the graph into graphs in
which all paths have finite length.  Nested words and the unfoldings
of finite directed graphs with bounded feedback vertex sets are
special instances of such graphs.

Our work is carried out in an automata-theoretic setting.  Roughly
speaking, the question of whether the alternation hierarchy collapses
to the alternation-free fragment over a class of structures~$\bfU$ can
be answered positively by showing that alternating parity automata are
as expressive as weak alternating automata over $\bfU$.  Translations
between automata and \mucalc formulas are known,
e.g.,~\cite{Niwinski1988:fixedpoints,Emerson_Jutla:focs1991,Kupferman_Vardi2005:linear_branching,Wilke:2001}.
Yet, the translation from weak alternating automata to
alternation-free formulas we provide here is more direct than the
known ones in the sense that it avoids the construction of formulas in
\mbox{vectorial form, cf.~\cite{an01}.}

Another technical contribution of this paper is a generalization of
the ranking construction developed by Kupferman and
Vardi~\cite{Kupferman_Vardi2001:weak}, which can be used to translate
alternating coB\"uchi word automata into language-equivalent weak
alternating word automata. We generalize it to the parity acceptance
condition and to more complex classes of structures, namely, to bounded
bottlenecked graphs.
Kupferman and Vardi~\cite{Kupferman_Vardi1998:weak_tree} have already
generalized their ranking construction for word automata and applied
it to solve the nonemptiness problem for nondeterministic parity tree
automata. However, our generalization of their ranking
construction~\cite{Kupferman_Vardi2001:weak} is conceptually simpler:
It eliminates the odd colors of a parity automaton in a single
construction step. An additional step is needed to obtain from the
resulting B\"uchi automaton a weak automaton.  In contrast, Kupferman
and Vardi's generalization~\cite{Kupferman_Vardi1998:weak_tree}
successively eliminates the colors, alternating between odd and even
colors. The acceptance conditions of the intermediate word automata
are a combination of a parity acceptance condition and a \mbox{B\"uchi or
coB\"uchi acceptance condition.}

We proceed as follows.  Preliminaries on \mucalc and alternating
automata are given in Section~\ref{sec:preliminaries}.  Translations
between \mucalc\!formulas and automata appear in
Section~\ref{sec:translations}.  Section~\ref{sec:constructions}
presents our generalization of the ranking construction.
Section~\ref{sec:bounded_connectivity} contains our collapse results.
Finally, in Section~\ref{sec:conclusion}, we draw conclusions and
outline directions for future work.  
Due to space limitations some of the proof details have been
omitted. They can be found in the full version of the paper, which is
\mbox{available from the authors' web pages.}

%%% Local Variables:
%%% mode: latex
%%% TeX-master: "main-eptcs-proc"
%%% End:

%% file: prelim-proc.tex
\section{Preliminaries}
\label{sec:preliminaries}

In this section, we provide notation and terminology that we use
throughout the paper.

\subsection{The $\mu$-Calculus}
\label{subsec:calculus}

\paragraph{Graphs}

Let $A$ be a nonempty finite set, whose elements are called
\emph{actions}, and let $\Sigma$ be an alphabet.  A
\emph{$(\Sigma,A)$-graph} is a directed, labeled, and pointed graph
$\big(V,(E_a)_{a\in A},v_{\rmI},\lambda\big)$, where $V$ is a set of
vertices, $E_a\subseteq V\times V$ is a set of edges labeled by $a\in
A$, $v_\rmI\in V$ is the source, and $\lambda:V\rightarrow\Sigma$ a
labeling function.  We require in the following that $V$ is at most
countable.

\paragraph{Syntax and Semantics}

We define the $\mu$-calculus, \mucalc for short, over
$(2^\propset,A)$-graphs, where $\propset$ is a nonempty set of
propositions.  Let $\varset=\set{X,Y,\dots}$ be a countable set of
variables.  The syntax of \mucalc is given by the grammar
\begin{equation*}
  \phi\mathbin{::=}
  X \mathbin{\big|} p \mathbin{\big|}  \neg p\mathbin{\big|} 
  \phi\wedge\phi \mathbin{\big|} \phi\vee\phi \mathbin{\big|} 
  \BOX{a}\phi \mathbin{\big|}  \DIAMOND{a}\phi \mathbin{\big|} 
  \mu X.\,\phi \mathbin{\big|} \nu X.\,\phi
  \,,
\end{equation*}
where $X$ ranges over $\varset$, $p$ over $\propset$, and $a$ over
$A$.
The semantics of \mucalc is as follows.  Let $\scrG=\big(V,(E_a)_{a\in
  A},v_{\rmI},\lambda\big)$ be a $(2^{\propset},A)$-graph.  A
valuation $\sigma$ assigns each variable in $\varset$ to a set of
vertices.  For $X\in\varset$ and $U\subseteq V$, we write
$\sigma[X\mapsto U]$ if we alter $\sigma$ at $X$, i.e.,
$\sigma[X\mapsto U](Y):=U$ if $Y=X$ and $\sigma[X\mapsto
U](Y):=\sigma(Y)$, otherwise.  The set $\sem{\phi}^\scrG_\sigma$ of
vertices in $\scrG$ that satisfy $\phi$ under $\sigma$ is
\mbox{defined as follows:}
\begin{align*}
   \sem{X}^\scrG_\sigma &:= \sigma(X) 
   \\%[.1cm]
   \sem{p}^\scrG_\sigma &:= \Setx{v\in V}{p\in \lambda(v)}
   \\%[.1cm]
   \sem{\neg p}^\scrG_\sigma &:= \Setx{v\in V}{p\not\in \lambda(v)}
   \\%[.1cm]
   \sem{\phi\wedge\psi}^\scrG_\sigma &:= 
   \sem{\phi}^\scrG_\sigma \cap \sem{\psi}^\scrG_\sigma
   \\%[.1cm]
   \sem{\phi\vee\psi}^\scrG_\sigma &:= 
   \sem{\phi}^\scrG_\sigma \cup \sem{\psi}^\scrG_\sigma
   \\%[.1cm]
   \sem{\BOX{a}\phi}^\scrG_\sigma &:=
   \Setx{v\in V}{\text{if }(v,v')\in E_a\text{ then }v'\in\sem{\phi}^\scrG_\sigma,
     \text{ for all }v'\in V}
   \\%[.1cm]
   \sem{\DIAMOND{a}\phi}^\scrG_\sigma &:=
   \Setx{v\in V}{(v,v')\in E_a\text{ and }v'\in\sem{\phi}^\scrG_\sigma,
     \text{ for some }v'\in V}
   \\%[.1cm]
   \sem{\mu X.\,\phi}^\scrG_\sigma &:=
   \bigcap\Setx{U\in 2^V}{\sem{\phi}^\scrG_{\sigma[X\mapsto U]}\subseteq U}
   \\%[.1cm]
   \sem{\nu X.\,\phi}^\scrG_\sigma &:=
   \bigcup\Setx{U\in 2^V}{\sem{\phi}^\scrG_{\sigma[X\mapsto U]}\supseteq U}
\end{align*}

The grammar of \mucalc guarantees that formulas are in negation normal
form, i.e., negations only occur directly in front of propositions in
$\propset$. This syntactic feature ensures monotonicity and thus
existence of the least and greatest fixpoints expressed by $\mu$ and
$\nu$, respectively.

The \emph{size} of $\phi$, written $|\phi|$, is
its number of syntactically distinct subformulas.
A formula $\phi$ is a \emph{sentence} iff $\phi$ does not have free
variables. In this case, $\sem{\phi}^\scrG_\sigma$ does not depend on 
$\sigma$. For a sentence $\phi$ and a set of
$(2^{\propset},A)$-graphs $\bfU$, we define
\begin{equation*}
  L_\bfU(\phi):=
  \Setx{\scrG\in\bfU}{v_{\rmI}\in\sem{\phi}^\scrG_\sigma,\text{
      with $v_{\rmI}$ the source of $\scrG$ and $\sigma$ some valuation}}
  \,.
\end{equation*}

\paragraph{Alternation Hierarchy}

\mucalc formulas determine an infinitely large hierarchy, which relies
on the mutual interdependencies between least and greatest fixpoint
operators. To define this hierarchy, we follow
Niwi{\'n}ski\cite{Niwinski1986:mucalculus}:
\begin{itemize}[--]
\item $\Sigma_0 = \Pi_0$ is the set of formulas without fixpoint
  operators, i.e., modal logic formulas.
\item For $n\geq0$, $\Sigma_{n+1}$ is the smallest set that contains
  the formulas in $\Sigma_n\cup\Pi_n$ and is closed under the
  following rules: (i)~if $\phi,\psi\in\Sigma_{n+1}$ then
  $\phi\wedge\psi\in\Sigma_{n+1}$ and $\phi\vee\psi\in\Sigma_{n+1}$;
  (ii)~if $\phi\in\Sigma_{n+1}$ and $a\in A$ then
  $\BOX{a}\phi\in\Sigma_{n+1}$ and $\DIAMOND{a}\phi\in\Sigma_{n+1}$;
  (iii)~if $\phi\in\Sigma_{n+1}$ and $X\in\varset$ then $\mu
  X.\,\phi\in\Sigma_{n+1}$; (iv)~if $\phi,\psi\in\Sigma_{n+1}$ and
  $X\in\varset$ then $\phi[\psi/X]\in\Sigma_{n+1}$ provided that no free
  variable of $\psi$ gets bound by a fixpoint operator in $\phi$ and
  where $\phi[\psi/X]$ denotes the formula obtained from substituting
  the free occurrences of $X$ by $\psi$ in $\phi$.
\item For $n\geq 0$, $\Pi_{n+1}$ is analogously defined as
  $\Sigma_{n+1}$: instead of closure under the least fixpoint operator
  $\mu$, we require closure with respect to the greatest fixpoint
  operator $\nu$.
\item For $n\geq 0$, we also define $\Delta_n:=\Sigma_n \cap \Pi_n$.
\end{itemize}
The \emph{alternation depth} of $\phi$, denoted by $\ad{\phi}$, is the
smallest $n\geq 0$ such that $\phi \in \Delta_{n+1}$.  A formula
$\phi$ is \emph{alternation-free} iff $\ad{\phi} \leq 1$, i.e., it is
in $\Delta_2$.

We remark that there is no agreement in the literature how to define
the alternation hierarchy of \mucalc. For instance, Emerson and
Lei~\cite{Emerson_Lei:lics1986} define $\Sigma_n$ and $\Pi_n$ slightly
differently. The differences are insubstantial for our results.
Furthermore, we point out that our definition of the alternation depth
of a formula is purely based on the formula's syntax and not on the
property it describes.

\subsection{Alternating Automata}
\label{subsec:automata}

\paragraph{Propositional Logic}

We denote the set of \emph{positive Boolean formulas} over the
proposition set $\propset$ by $\Bool^+(\propset)$, i.e.,
$\Bool^+(\propset)$ consists of the formulas that are inductively
built from the Boolean constants $\TRUE$ and $\FALSE$, the
propositions in $\propset$, and the Boolean connectives $\vee$ and
$\wedge$.  For $\calM \subseteq \propset$ and $\phi \in
\Bool^+(\propset)$, write $\calM \models \phi$ iff $\phi$ holds when
assigning true to the propositions in $\calM$ and false to those in
$\propset \setminus \calM$.

\paragraph{Words and Trees}

We denote the set of finite words over the alphabet $\Sigma$ by
$\Sigma^*$, the set of infinite words over $\Sigma$ by
$\Sigma^\omega$, and the empty word by $\epsilon$.
For a word $w$, $w_i$ denotes the symbol of $w$ at position $(i+1)$.
Write $v\preceq w$ if $v$ is a prefix of $w$.

A ($\Sigma$-labeled) \emph{tree} is a function $t:T\rightarrow
\Sigma$, where $T\subseteq\Nat^*$ satisfies the following conditions:
(i)~$T$ is prefix-closed, i.e., $v\in T$ and $u\preceq v$ implies
$u\in T$, and (ii)~if $vi\in T$ and $i>0$ then $v(i-1)\in T$.
The elements in $T$ are called the \emph{nodes} of $t$ and the empty
word $\epsilon$ is called the \emph{root} of $t$.  A node $vi\in T$
with $i\in\Nat$ is called a \emph{child} of the node $v\in T$.
A \emph{branch} in $t$ is a word $\pi\in\Nat^*\cup\Nat^{\omega}$ such
that either $\pi\in T$ and $\pi$ does not have any children, or $\pi$
is infinite and every finite prefix of $\pi$ is in $T$.  We write
$\bar{t}(\pi)$ for the word $t(\epsilon) t(\pi_0) t(\pi_0\pi_1)\ldots
t(\pi_0\pi_1\dots\pi_{n-1})\in\Sigma^*$ if $\pi$ is a finite branch of
length $n$ and $t(\epsilon) t(\pi_0)
t(\pi_0\pi_1)\ldots\in\Sigma^{\omega}$ if $\pi$ is infinite.

\paragraph{Automata}

In the following, we define alternating automata where the inputs are
$(2^\propset,A)$-graphs, where $\propset$ is a nonempty finite set of
propositions and $A$ is a nonempty finite set of actions.  Such
automata are essentially alternating parity tree automata that operate
over the tree unfolding of the given input.  The classical automata
models for words and trees are special instances when encoding the
letters of an alphabet $\Sigma$ by subsets of propositions and by
viewing words and trees in a \mbox{straightforward way as
$(\Sigma,A)$-graphs.}

A \emph{parity $(\propset,A)$-automaton}, $(\propset,A)$-PA for short,
is a tuple $\autA=(Q,\delta,q_\rmI,\alpha)$, where $Q$ is a finite set
of states, $\delta: Q\rightarrow
\Bool^+\big(\propset\cup\bar{\propset}\cup (Q \times
\set{\Diamond,\Box}\times A)\big)$ is the transition function with
$\bar{\propset}:=\setx{\bar{p}}{p\in\propset}$, $q_\rmI \in Q$ is the
initial state, and $\alpha:Q\rightarrow\Nat$ determines the (parity)
acceptance condition.  Assume that
$\propset\cap\bar{\propset}=\emptyset$.  We refer to $\alpha(q)$ as
the \emph{color} of the state $q\in Q$.  The \emph{index} of $\autA$
is $\ind{\autA} := |\setx{\alpha(q)}{q \in Q}|$ and the \emph{size} of
$\autA$ is the number of syntactically distinct subformulas that occur
in the transitions, i.e., $\size{\autA}:=\big|\bigcup_{q\in
  Q}\Setx{\psi}{\psi\text{ is a subformula of }\delta(q)}\big|$.  In
the following, we assume that $|Q|\in\bigOh(\size{\autA})$, which
holds when, e.g., every state occurs in some transition of $\autA$.

Let $\autA=(Q,\delta,q_\rmI,\alpha)$ be a $(\propset,A)$-PA and
$\scrG=\big(V,(E_a)_{a\in A}, v_\rmI,\lambda\big)$ a
$(2^\propset,A)$-graph.  A \emph{run} of $\autA$ on $\scrG$ is a tree
$\rho: R \to V\times Q$ with some $R\subseteq \Nat^*$ such that
$\rho(\epsilon) = (v_\rmI,q_\rmI)$ and for each node $x \in R$ with
$\rho(x) = (v,p)$, there is a set $\calM\subseteq
Q\times\{\Diamond,\Box\}\times A$ such that
\begin{equation*}
  \Setx{q\in\propset}{q\in\lambda(v)}\cup
  \Setx{\bar{q}\in\bar{\propset}}{q\not\in\lambda(v)}\cup
  \calM \models\delta(p)
\end{equation*}
and the following conditions are satisfied:
\begin{enumerate}[(a)]
\item If $(q,\Diamond,a)\in \calM$, then there is a node $v'\in V$ with
  $(v,v')\in E_a$ such that there is a child $x'\in R$ of $x$ with
  $\rho(x') = (v',q)$.
\item If $(q,\Box,a)\in\calM$, then for all nodes $v'\in V$ with
  $(v,v')\in E_a$ there is a child $x'\in R$ of $x$ such that
  $\rho(x') = (v',q)$.
\end{enumerate}
Roughly speaking, $\autA$ starts in its initial state by scanning the
input graph from its source. The label $(v,p)$ of the node $x$ in the
run is the current configuration of $\autA$. That is, $\autA$ is
currently in the state $p$ and the read-only head is at the vertex $v$
of the input.  The transition $\delta(p)$ specifies with respect to
the labeling $\lambda(v)$ a constraint that has to be respected by the
automaton's successor states.  In particular, for a proposition
$(q,\Diamond, a)\in\calM$, the read-only head must move along some
$a$-labeled edge starting at $v$. Similarly, for $(q,\Box,a)\in\calM$,
a copy of the read-only head must move along every $a$-labeled edge
starting at vertex~$v$.

An infinite branch $\pi$ in a run $\rho$ with $\bar{\rho}(\pi) =
(v_0,q_0) (v_1,q_1) \ldots$ is \emph{accepting} iff
$\max\setx{\alpha(q)}{q\in \inf(q_0q_1 \ldots)}$ is even, where
$\inf(q_0q_1\dots)$ denotes the set of states that occur infinitely
often in $q_0q_1\dots$.  The run $\rho$ is \emph{accepting} iff every
infinite branch in $\rho$ is accepting.
The \emph{language} of $\autA$ with respect to a set $\bfU$ of
$(2^\propset,A)$-graphs is the set
\begin{equation*}
  L_\bfU(\autA) :=
  \setx{\scrG\in\bfU} 
  {\text{there is an accepting run of $\autA$ on $\scrG$}}
  \,.
\end{equation*}

By restricting the acceptance condition and the automaton's
transitions, we obtain the following automata classes. Let
$\autA=(Q,\delta,q_{\rmI},\alpha)$ be a $(\propset,A)$-PA.
\begin{itemize}[--]
\item $\autA$ is \emph{B\"uchi} iff $\setx{\alpha(q)}{q\in
    Q}\subseteq\{1,2\}$.
\item $\autA$ is \emph{coB\"uchi} iff $\setx{\alpha(q)}{q\in
    Q}\subseteq\{0,1\}$.
\item $\autA$ is \emph{weak} iff there is a partition $Q_0,\dots,Q_n$
  on $Q$, for some $n\geq 0$ such that for all $i\in\{0,\dots,n\}$,
  the following holds: (i)~All states in the component $Q_i$ have the
  same parity, i.e.,~${\alpha(q)} \equiv {{\alpha(q')}\!\!\mod{2}}$, for
  all $q,q'\in Q_i$.  (ii)~$\delta(q)\in
  \Bool^+\big(\propset\cup\bar{\propset}\cup
  \bigcup_{j\in\set{i,\dots,n}}(Q_j\times\set{\Diamond,\Box}\times
  A)\big)$, for all $q\in Q_i$.  That is, when reading a vertex label
  the automaton can stay in the current component $Q_i$ or go to
  components with higher indices.
\end{itemize}
We also call $\autA$ a $(\propset,A)$-BA, $(\propset,A)$-CA, and
$(\propset,A)$-WA when it is B\"uchi, coB\"uchi, and weak,
respectively.

Finally, dualizing an alternating automaton corresponds to
complementation~\cite{Muller1987}.  In our case, the \emph{dual
  automaton} of a $(\propset,A)$-PA $\autA = (Q,\delta,q_I,\alpha)$ is
defined as the $(\propset,A)$-PA
$\overline{\autA}:=(Q,\overline{\delta},q_I,\overline{\alpha})$, where
for each $q \in Q$, $\overline{\delta}(q):=\overline{\delta(q)}$ with
\begin{align*}
  \overline{\TRUE} &:= \FALSE &\qquad
  \overline{\FALSE} &:= \TRUE \\
  \overline{p} &:= \bar{p}, \text{ for } p \in \propset &\qquad
  \overline{\bar{p}} &:= p, \text{ for } \bar{p} \in \bar{\propset} \\
  \overline{(q,\Diamond,a)} &:= (q,\Box,a) &\qquad
  \overline{(q,\Box,a)} &:= (q,\Diamond,a) \\
  \overline{\beta \wedge \gamma} &:= \overline{\beta} \vee \overline{\gamma} &\qquad
  \overline{\beta \vee \gamma} &:= \overline{\beta} \wedge \overline{\gamma}
\end{align*}
and $\overline{\alpha}(q) := \alpha(q)+1$.
It is not too hard to show that the dual automaton accepts the
complement language, i.e., $L_{\bfU}(\overline{\autA}) =
\overline{L_{\bfU}(\autA)}$.  Furthermore, note that
$\overline{\autA}$ is weak if $\autA$ is weak.

%%% Local Variables:
%%% mode: latex
%%% TeX-master: "main-eptcs-proc"
%%% End:

%% file: trans-proc.tex
\section{From the $\mu$-Calculus to Automata and Back}
\label{sec:translations}

Translations between \mucalc and automata are known for various
automaton models.  For the sake of completeness, we present in this
section such translations with respect to our automaton model from
Section~\ref{subsec:automata}.
In the remainder of the text, let $\propset$ and $A$ be nonempty
finite sets of propositions and actions, respectively.  Furthermore,
throughout this section, let $\bfU$ be a set of
$(2^{\propset},A)$-graphs

\subsection{From \mucalc to Parity Automata}

The following translation is similar to the one 
in~\cite{Wilke:2001}. However, since our automaton
model does not support $\epsilon$-transitions, we need to require for
the translation that formulas are guarded, i.e., variables occur under
the scope of a modal operator within their defining fixpoint
formulas. For a proof of the following lemma, see,
e.g.,~\cite{Wal:comkap}.
\begin{lemma}
  \label{lem:guarded}
  For every sentence $\phi$, there is a guarded sentence $\psi$ of
  size $2^{\bigOh(|\phi|)}$ such that $L_\bfU(\psi)=L_\bfU(\phi)$ and
  $\ad{\psi}=\ad{\phi}$.\footnote{We are not aware of polynomial
    translations into the guarded fragment.  The claimed polynomial
    upper bounds of translations found in the literature are flawed.
    Counterexamples are families of formulas like $\mu X_1.\ldots\mu
    X_n.\,\bigvee_{i=1}^n X_i \wedge \DIAMOND{a}\bigwedge_{i=1}^n
    X_i$, where $a$ is an action.  For the given translations, these
    formulas cause exponential blow-ups.}
\end{lemma}

From guarded formulas one easily obtains equivalent parity automata.
\begin{theorem}
  \label{thm:mucalc_APA}
  For every guarded sentence $\phi$, there is a $(\propset,A)$-PA
  $\autA_\phi$ with $|\phi|$ states and
  $L_{\bfU}(\autA_\phi)=L_\bfU(\phi)$.  Moreover,
  $\size{\autA_\phi}\in\bigOh(|\phi|)$ and $\ind{\autA_\phi} \le
  \ad{\phi}+1$.
\end{theorem}

\subsection{From Weak Automata to Alternation-free \mucalc}

The following translation is similar to the one 
in~\cite{Kupferman_Vardi2005:linear_branching}.
However, our variant avoids a vectorial form for \mucalc formulas.

Let $\autA = (Q,\delta,q_\rmI,\alpha)$ be a $(\propset,A)$-WA.
Without loss of generality, assume $Q = \{q_0,\ldots,q_n\}$,
$q_{\rmI}=q_0$, and for $i,j\in\{0,\dots,n\}$, if $q_j$ occurs in the
Boolean formula $\delta(q_i)$ then $i<j$ or $\alpha(q_i)\equiv
\alpha(q_j)\!\!\mod 2$. Also, assume that the Boolean constants
$\TRUE$ and $\FALSE$ do not occur in $\autA$'s transitions.

From $\autA$ we define the \mucalc sentence $\phi_{\autA}$ with the
variables $X_0,\ldots,X_n\in\varset$.  Intuitively, $X_i$ evaluates to
the set of vertices of an input that can be labeled by $q_i$ in an
accepting run. We obtain the formula $\varphi_{\autA}$ from the
formulas $\psi_n,\dots,\psi_0$ defined inductively for $i=n,\ldots,0$:
let $\psi_i:=\kappa_i X_i.\,\mathit{tr}_i\big(\delta(q_i)\big)$, where
$\kappa_i:=\mu$ if $\alpha(q_i)$ is odd and $\kappa_i:=\nu$ if
$\alpha(q_i)$ is even, and the function $\mathit{tr_i}$ is as follows:
\begin{equation*}
 \mathit{tr}_i(\phi) :=
 \begin{cases}
   p 
   &\text{if $\phi=p$ with $p\in\propset$}
   \\
   \neg p 
   &\text{if $\phi=\bar{p}$ with $\bar{p}\in\bar{\propset}$}
   \\
   \DIAMOND{a}\psi_j
   &\text{if $\phi=(q_j,\Diamond,a)$ and $j>i$}
   \\
   \DIAMOND{a}X_j
   &\text{if $\phi=(q_j,\Diamond,a)$ and $j\leq i$}
   \\
   \BOX{a}\psi_j
   &\text{if $\phi=(q_j,\Box,a)$ and $j>i$}
   \\
   \BOX{a}X_j
   &\text{if $\phi=(q_j,\Box,a)$ and $j\leq i$}
   \\
   \mathit{tr}_i(\psi) \star \mathit{tr}_i(\psi') 
   &\text{if $\phi=\psi\star\psi'$ with $\star\in\{\wedge,\vee\}$}
 \end{cases}
\end{equation*}
We point out that at most the variables $X_0,\ldots,X_{i-1}$ occur
free in $\psi_i$.  With $\phi_\autA:=\psi_0$ we obtain the following
theorem.  
\begin{theorem}
  \label{thm:weakABA_altfreemucalc}
  For every $(\propset,A)$-WA $\autA$ with $n$ states, there is an
  alternation-free sentence $\phi_\autA$ with
  $|\phi_\autA|\in\bigOh\big(n\cdot\size{\autA}\big)$ and
  $L_{\bfU}(\phi_\autA)=L_\bfU(\autA)$.
\end{theorem}

%%% Local Variables:
%%% mode: latex
%%% TeX-master: "main-eptcs-proc"
%%% End:

%% file: constr-proc.tex
\section{From Parity Automata to Weak Automata}
\label{sec:constructions}

In this section, we show that parity automata and weak automata have
the same expressive power over so-called {\em bottlenecked directed
  acyclic graphs} (BDAGs) with a bounded width.  BDAGs are fundamental
to this paper as our collapse results rely on reductions of different
structures---such as various classes of graphs and words---to BDAGs
with a bounded width; see Section~\ref{sec:bounded_connectivity}.
The schematic form of BDAGs is illustrated in
Figure~\ref{fig:bottleneckedDAGs}. 
\begin{figure}[t]
  \vspace{-.3cm}
  \begin{center}
    \includegraphics[scale=.6]{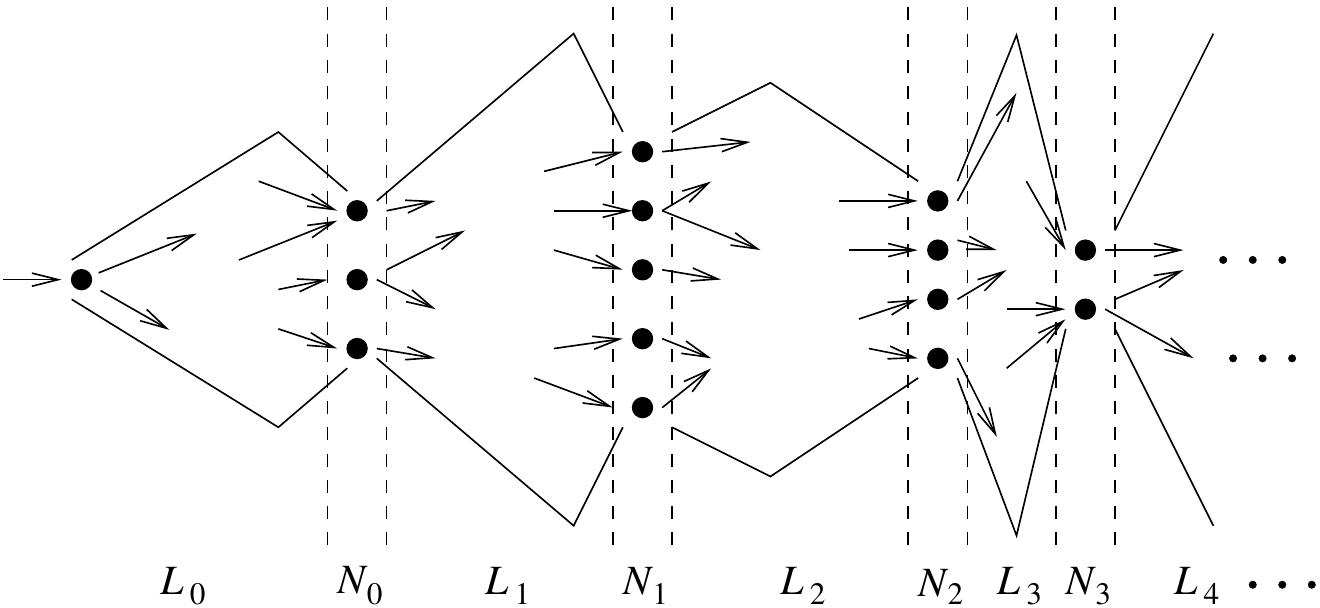}
  \end{center}
  \vspace{-.3cm}
  \caption{\small \label{fig:bottleneckedDAGs}Bottlenecked directed
    acyclic graph (BDAG)}
  \vspace{-.2cm}
\end{figure}
Their definition is as follows.
\begin{definition}
  \label{def:bottleneckedDAG}
  Let $\scrG=\big(V,(E_a)_{a\in A},v_\rmI,\lambda\big)$ be a
  $(\Sigma,A)$-graph.
  \begin{itemize}[--]
  \item $\scrG$ is a \emph{directed acyclic graph} (DAG) iff it does
    not contain cycles, i.e., there are no vertices $v_0,\dots,v_n\in
    V$ with $n\geq 1$ such that $v_0=v_n$ and
    $(v_i,v_{i+1})\in\bigcup_{a\in A}E_a$, for all $i\in\Nat$ with
    $0\leq i<n$.
  \item $\scrG$ is a \emph{bottlenecked DAG} (BDAG) of width
    $w\in\Nat$ iff $\scrG$ is a DAG and $V$ can be split into the
    pairwise disjoint sets $L_0,N_0,L_1,N_1,\dots$ such that
    \begin{enumerate}[(i)]
    \item $\bigcup_{a\in A} E_a\subseteq {\bigcup_{i\in\Nat}}
      \big((L_i\times L_i)\cup (L_i\times N_i)\cup (N_i\times
      L_{i+1})\big)$, 
    \item $w=\sup\Setx{|N_i|}{i\in\Nat}$, and
    \item each $L_i$ is well-founded, i.e., the graph obtained from
      $\scrG$ by restricting the vertex set to $L_i$ does not contain
      infinite paths.
    \end{enumerate}    
  \end{itemize}
\end{definition}
Note that BDAGs naturally define a {\em connectivity} measure, which
is given by their widths: removing the vertices in the $N_i$s
disconnects the structure into DAGs in which all paths are finite and
thus the infinite behavior described by the original structure is
eliminated.

Before presenting our collapse results in
Section~\ref{sec:bounded_connectivity}, we need the following
construction, parametric in $w\in\Nat$, that translates parity
automata into language-equivalent weak automata with respect to the
class of bottlenecked graphs of width at most $w$.
In the following, let $w\in\Nat$ and let $\BDAG_{\leq w}$ be the class
of $(2^\propset,A)$-graphs that are BDAGs of width at most $w$.
Moreover, for $n\in\Nat$, we abbreviate the set
  $\set{0,1,\dots,n}$ by $[n]$.

\subsection{Rankings}

Let $\autA = (Q, \delta, q_\rmI, \alpha)$ be a $(\propset,A)$-PA and
$\rho: R \to V\times Q$ a run of $\autA$ on $\scrG\in\BDAG_{\leq w}$
with $\scrG=\big(V,(E_a)_{a\in A},v_{\rmI}, \lambda\big)$.  Without
loss of generality, we assume that in the run $\rho$ equally labeled
nodes have isomorphic subtrees and therefore that $\rho$ is
\emph{memoryless}; formally, for all $x,y\in R$ if $\rho(x)=\rho(y)$
then for all $z\in\Nat^*$, whenever $xz\in R$ then $yz\in R$ and
$\rho(xz)=\rho(yz)$.
For the memoryless run $\rho$, we define the graph
$G^\rho:=(V^\rho,E^\rho)$ with $V^\rho := \Setx{\rho(x)}{\text{$x \in
    R$}}$ and $ E^\rho := \Setx{\big(\rho(x), \rho(y)\big)}
{\text{$x,y \in R$ and $y$ is a child of $x$}}$.  The graph $G^\rho$
is a representation of the memoryless run $\rho$ in which equally
labeled nodes are merged. 
Furthermore, $G^{\rho}$ is a BDAG of width at most $|Q|w$. 

Let $c\geq0$.  A state $q\in Q$ is \emph{$c$-releasing} iff
$\alpha(q)>c$ and $\alpha(q)\not\equiv c\!\!\mod 2$. 
An infinite path of the form $(h_0,q_0) (h_1,q_1) \dots$ in $G^\rho$
is \emph{$c$-dominated} iff there is a state $q\in\inf(q_0q_1\dots)$
with $\alpha(q)=c$ and no $c$-releasing state in $\inf(q_0q_1\dots)$.
The function $f:V^\rho\rightarrow[2|Q|w]$ is a
\emph{$c$-ranking} for $G^\rho$ iff the following two conditions hold:
\begin{enumerate}[(i)]
\item For all $(h,q) \in V^\rho$, if $f(h,q)$ is odd then $\alpha(q)\not=c$.
\item For all $v,v'\in V^\rho$ with $v=(h,q)$, if $(v,v')\in E^\rho$
  and $f(v)<f(v')$ then $q$ is $c$-releasing.
\end{enumerate}
The $c$-ranking $f$ is \emph{safe} iff every infinite path in $G^\rho$
either visits infinitely many vertices with $c$-releasing states or
$f$ gets trapped in an odd rank on the path, i.e., iff for every
infinite path $(h_0,q_0)(h_1,q_1) \dots$ in $G^\rho$, either there is a
state $q\in\inf(q_0q_1\dots)$ with $\alpha(q)>c$ and
$\alpha(q)\not\equiv c\!\!\mod 2$, or there is an integer $n\in \Nat$ such
that $f(h_n,q_n)$ is odd and $f(h_j,q_j) = f(h_n,q_n)$, for all $j\ge
n$.
We point out that the color $\alpha(q)$ of a state $q\in Q$ and the
rank $f(h,q)$ of a vertex $(h,q)\in V^\rho$ have different
meanings. In particular, the parities of $\alpha(q)$ and $f(h,q)$ can
differ.

It holds that the run $\rho$ is accepting iff for all odd $c\geq1$,
all infinite paths in $G^{\rho}$ are not $c$-dominated.  The following
theorem reduces the problem of checking whether every infinite path in
$G^\rho$ is not dominated by one specific color to the problem of
checking the existence of a safe ranking for $G^\rho$.  
\begin{theorem}
  \label{thm:ranking}
  Let $c\geq0$.  Every infinite path in $G^\rho$ is not $c$-dominated
  iff there is a safe $c$-ranking for $G^\rho$.
\end{theorem}

The proof of Theorem~\ref{thm:ranking} is based on ingredients that
appear in the Kupferman and Vardi's correctness proof of the
construction that translates alternating coB\"uchi word automata into
weak alternating word automata~\cite{Kupferman_Vardi2001:weak}.  Since
our automata are parity automata that operate over BDAGs instead of
words, some arguments are more subtle than in the
coB\"uchi-word-automata case.

In the following, we show that the existence of a safe ranking can be
checked by a B\"uchi automaton. The ranks are guessed during a run
with the states of the B\"uchi automaton.  The conditions~(i) and~(ii)
of a ranking are locally checked by the transition function of the
automaton. With the acceptance condition of the automaton we check
whether the guessed ranking is safe.  Details of the construction are
given in Theorem~\ref{thm:color_elim} below.  For proving the
correctness of the construction, it does not suffice to only assume
the existence of a safe ranking. The ranking must also satisfy
additional technical requirements, which are guaranteed by the
following lemma. 
\begin{lemma}
  \label{lem:normalized_ranking}
  Let $c\geq0$.  If $G^\rho$ has a safe $c$-ranking then there is a
  safe $c$-ranking $g:V^\rho\rightarrow [2|Q|w]$ that
  satisfies the following additional properties:
  \begin{itemize}[--]
  \item $g(v_\rmI,q_\rmI)=2|Q|w$, and 
  \item $g(h_1,q)=g(h_2,q)$, for all vertices $(h_1,q),(h_2,q)\in
    V^{\rho}$ for which there exists a vertex $(h',p)\in V^{\rho}$
    with $\big((h',p),(h_1,q)\big)\in E^{\rho}$ and
    $\big((h',p),(h_2,q)\big)\in E^{\rho}$.
  \end{itemize}
\end{lemma}

We finally present the construction of the B\"uchi automaton that
checks whether a safe ranking exists.
\begin{theorem}
  \label{thm:color_elim}
  Let $c\geq0$.  There is a $(\propset,A)$-BA $\autB_c$ with
  $|Q|\cdot(2|Q|w+1)$ states and $L_{\BDAG_{\leq w}}(\autB_c)$ equals
  \begin{equation*}
    \Setx{\scrG\in\BDAG_{\leq w}}
    {\text{there is a memoryless run $\rho$ of
        $\autA$ on $\scrG$ such that $G^\rho$ has a safe $c$-ranking}}
    \,.
  \end{equation*}
  Furthermore,
  $\size{\autB_c}\in\bigOh\big(\size{\autA}\cdot(|Q|w+1)\big)$.
\end{theorem}
\begin{proof}
  We define $\autB_c$ as $\big(Q\times[2|Q|w], \eta,
  p_\rmI, \beta\big)$, where $p_\rmI$, $\eta$, and $\beta$ are as
  follows:
  \begin{itemize}[--]
  \item The initial state $p_\rmI$ is the tuple $(q_\rmI, 2|Q|w)$.
  \item To define the transition function $\eta$, we need the
    following two definitions.
    (1)~For $q \in Q$ and $r, r' \in[2|Q|w]$, we write $r' \preceq_q
    r$ if either $r' \le r$ or $q$ is $c$-releasing.  (2)~For $\phi
    \in \Bool^+\big(\propset\cup\bar{\propset}\cup (Q \times
    \set{\Diamond,\Box}\times A)\big)$, $q \in Q$, and $r \in
    [2|Q|w]$, we define $\release_q(\phi, r)$ as the positive Boolean
    formula that we obtain by replacing each proposition $(p, \star,
    a)$ in $\phi$ by the disjunction $\bigvee_{r' \preceq_q r}
    \big((p, r'),\star,a\big)$.
    For $q \in Q$ and $r\in [2|Q|w]$, we define
    \begin{equation*}
      \eta(q, r) :=
      \begin{cases}
        \release_q\big(\delta(q), r\big) & 
        \text{if $\alpha(q)\not=c$ or $r$ is even,}
        \\
        \FALSE                    & \text{otherwise.}     
      \end{cases}
    \end{equation*}
  \item The acceptance condition is determined by 
    $\beta:Q\times[2|Q|w]\rightarrow \{1,2\}$ where 
    \begin{equation*}
      \beta(q,r):=
      \begin{cases}
        2 & \text{if $q$ is $c$-releasing or $r$ is odd,}
        \\
        1 & \text{otherwise.}
      \end{cases}
    \end{equation*}
  \end{itemize}
  
  Obviously, $\autB_c$ has $|Q|\cdot(2|Q|w+1)$ states.  An upper bound
  on the number of distinct subformulas in the positive Boolean
  formula $\eta(q,r)$ for $q\in Q$ and $r\in[2|Q|w]$ is
  $\bigOh\big(m+|Q|\cdot(|Q|w+1)\big)$, where $m$ is the number of
  distinct formulas in $\delta(q)$.  Note that the disjunction
  $\bigvee_{r'\preceq_q r}\big((p,r'),\star,a\big)$ in $\eta(q,r)$,
  which replaces a proposition of the form $(p,\star,a)$ in
  $\delta(q)$, is a subformula of $\bigvee_{0\leq r'\leq 2|Q|w}
  \big((p,r'),\star,a)\big)$. The disjunction $\bigvee_{0\leq r'\leq
    2|Q|w}\big((p,r'),\star,a)\big)$ has $\bigOh(|Q|w+1)$ subformulas.
  Since we count multiple occurrences of the same subformula in the
  transitions of an automaton only once, we obtain that
  $\size{\autB_c}\in\bigOh\big(\size{\autA}\cdot(|Q|w+1)+|Q|\cdot(|Q|w+1)\big)
  \subseteq \bigOh\big(\size{\autA}\cdot(|Q|w+1)\big)$.
  It remains to prove that $\scrG\in L_{\BDAG_{\leq w}}(\autB_c)$ iff
  there is a run $\rho$ of $\autA$ on $\scrG$ such that the graph
  $G^\rho$ has a safe $c$-ranking.

  \medskip
  \noindent
  %\begin{CASE}
  {$(\Rightarrow)$}
    Let $\rho': R \to V\times \big(Q \times [2|Q|w]\big)$ be an
    accepting, memoryless run of $\autB_c$ on
    $\scrG=\big(V,(E_a)_{a\in A},v_{\rmI},\lambda\big)$.  We define
    the tree $\rho: R \to V\times Q$ with $\rho(x) := (h,q)$, for
    every $x \in R$ with $\rho'(x) = \big(h,(q, r)\big)$, i.e., the
    labels of the nodes in $\rho$ are the projections of the labels of
    $\rho'$ on $V\times Q$.
    The tree $\rho$ is a run of $\autA$ on $\scrG$ since the
    transition function of $\autB_c$ just annotates state of $\autA$
    by ranks.  We can assume that there are no $x,y\in R$ with
    $\rho'(x)=(h,(q,r))$, $\rho'(y)=(h,(q,r))$, and $r\not=r'$. That
    is, the rank $r\in[2|Q|w]$ assigned by the run $\rho'$ to a vertex
    $(h,q)$ in the graph $G^\rho$ representing $\rho$ is unique. We
    define $f(h,q):=r$.

    It follows from the definition of $\eta$ that $f$ is a $c$-ranking
    for $G^{\rho}$.
    Since $\rho'$ is accepting, on every branch $\pi$ in $\rho'$ there
    are either $c$-releasing states or odd ranks which, in both cases,
    occur infinitely often.  The case where $\pi$ visits infinitely
    many vertices with $c$-releasing states is obvious. Assume that
    $\pi$ visits only finitely many vertices with $c$-releasing
    states. Then, the ranks do not increase from some point
    onwards. Thus, they must eventually stabilize. We conclude that
    $f$ is safe.
  %\end{CASE}

  \medskip
  \noindent
  %\begin{CASE}
  {$(\Leftarrow)$}
    Let $f:V^\rho\rightarrow[2|Q|w]$ be a safe $c$-ranking on the
    graph representation $G^{\rho}=(V^\rho,E^\rho)$ of the run $\rho:R
    \rightarrow V\times Q$ of $\autA$ on $\scrG=\big(V,(E_a)_{a\in
      A},v_{\rmI},\lambda\big)$.
    The idea is to attach the ranks given by $f$ to the labels of the
    nodes in $\rho$ to obtain an accepting run $\rho':R\to V\times
    \big(Q\times[2|Q|w]\big)$ of $\autB_c$ on $\scrG$. However, we
    cannot use $f$ directly, since the following situation might
    occur.  Assume that there are vertices $h,h_1,h_2\in V$ with
    $(h,h_1),(h,h_2)\in E_a$, for some $a\in A$. Furthermore, assume
    $\rho(x)=(h,p)$ and $\delta(p)=(q,\Box,d)$, for some node $x\in R$
    and states $p,q\in Q$. Then, the node $x$ must have children
    $y,y'\in R$ such that $\rho(y)=(h_1,q)$ and $\rho(y')=(h_2,q)$.
    If $\rho'$ attaches the ranks of $(h,p)$, $(h_1,q)$ and $(h_2,q)$
    to the labels of the nodes $x$, $y$, and $y'$, respectively, i.e.,
    $\rho'(x)=\big(p,f(h,p)\big)$, $\rho'(y)=\big(q,f(h_1,q)\big)$,
    $\rho'(y')=\big(q,(f(h_2,q)\big)$, we do not obtain a run when the
    ranks of $(h_1,q)$ and $(h_2,q)$ differ. However,
    Lemma~\ref{lem:normalized_ranking} allows us to assume that
    $f(h_1,q)=f(h_2,q)$.  In the following, let $f$ be a safe
    $c$-ranking with the additional \mbox{properties in
    Lemma~\ref{lem:normalized_ranking}.}

    We define the tree $\rho': R \to V\times \big(Q \times
    [2|Q|w]\big)$ now by $\rho'(x):=\big(h,(q,f(h,q))\big)$, for $x\in
    R$ with $\rho(x)=(h,q)$.
    We first show that $\rho'$ is a run of $\autB_c$ on $\scrG$.
    By definition of $\rho'$, we have $\rho'(\epsilon) =
    \big(v_{\rmI},(q_\rmI, 2|Q|w)\big)$.  Hence, the root of $\rho'$
    is well labeled with respect to the initial state of $\autB_c$.
    Consider a node $x \in R$ with $\rho(x) = (h,q)$ and assume that
    $r \in [2|Q|w]$ is the rank of $(h,q)$, i.e., $r=f(h,q)$.  Let $S$
    be the set of labels of the successors of node $x$ in $\rho$.  By
    condition~(ii) of a ranking, we have $f(h',q') \le r$, for each
    $(h',q)\in S$ if $q$ is not $c$-releasing.  Furthermore,
    $\alpha(q)=c$ and $r$ odd cannot hold at the same time because of
    condition~(i) of a ranking.  Moreover, by
    Lemma~\ref{lem:normalized_ranking}, we have $f(h',q')=f(h'',q'')$
    whenever $q'=q''$, for all $(h',q'),(h'',q'')\in S$.  Thus, the
    set $S'$ of the labels of the successor nodes of $x$ in $\rho'$ is
    $\Setx{(h',(q',r'))}{(h',q')\in S\text{ and }r'=f(h',q')}$.  It is
    now easy to obtain from this set $S'$ of labels a model of
    $\eta(q, r)$ which witnesses that the labeling corresponds to a
    valid transition of $\autB_c$ with respect to the vertex label
    $\lambda(h)$.

    The run $\rho'$ is accepting: since $f$ is safe, every infinite
    path in $G^{\rho}$ that does not visit $c$-releasing vertices
    infinitely often gets trapped in an odd rank.  Then, by the
    definition of $\beta$, every infinite branch in $\rho'$ is
    accepting.  
  %\end{CASE}
\end{proof}

\subsection{Applications}

The first application is to obtain weak automata from B\"uchi
automata.
\begin{lemma}
  \label{lem:BA_WA}
  Let $\autA$ be a $(\propset,A)$-BA with $n$ states. There is a
  $(\propset,A)$-WA $\autB$ with $n(2nw+1)$ states and $L_{\BDAG_{\leq
      w}}(\autB)=L_{\BDAG_{\leq w}}(\autA)$.  Furthermore,
  $\size{\autB}\in\bigOh\big(\size{\autA}\cdot(nw+1)\big)$.
\end{lemma}
\begin{proof}
  First construct from $\autA$ the coB\"uchi automaton $\autC$ by
  dualizing the transition function of $\autA$ and its acceptance
  condition. $\autC$ accepts the complement of $\autA$. Let $\autB_1$
  be the B\"uchi automaton obtained from Theorem~\ref{thm:color_elim}
  for the only odd color~$1$. This automaton is weak as $\autC$ does
  not have $1$-releasing states. It has $n(2nw+1)$ states and
  $\size{\autB_1}\in \bigOh\big(\size{\autA}\cdot(nw+1)\big)$.  It
  follows from Theorem~\ref{thm:ranking} that $L_{\BDAG_{\leq
      w}}(\autB_1)=L_{\BDAG_{\leq w}}(\autC)$.  The dual automaton of
  $\autB_1$ accepts $L_{\BDAG_{\leq w}}(\autA)$.
\end{proof}

We now show how to combine B\"uchi automata for different odd colors
from Theorem~\ref{thm:color_elim} so that they simultaneously check
the existence of safe rankings.
\begin{lemma}
  \label{lem:PA_BA}
  Let $\autA$ be a $(\propset,A)$-PA with $n$ states and index
  $k$. There is a $(\propset,A)$-BA $\autB$ with $\bigOh\big(kn
  (2nw+1)^{\lceil k/2\rceil}\big)$ states and $L_{\BDAG_{\leq
      w}}(\autB)=L_{\BDAG_{\leq w}}(\autA)$.  Moreover,
  $\size{\autB}\in\bigOh\big(k\size{\autA}(2nw+1)^{\lceil
    k/2\rceil}\big)$.
\end{lemma}
\begin{proof}
  Assume the odd colors of $\autA$ are $c_1,\dots,c_\ell\in\Nat$,
  for some $\ell\leq \lceil k/2\rceil$.  For $i\in\{1,\dots,\ell\}$,
  let $\autB_{c_i}$ be the B\"uchi automaton from
  Theorem~\ref{thm:color_elim}.  From the automata
  $\autB_{c_1},\dots,\autB_{c_\ell}$, we construct a so-called
  generalized B\"uchi automaton $\autC$, 
  i.e., one where the 
  acceptance condition is the finite conjunction of finitely many
  B\"uchi acceptance conditions.
  Since the transition functions
  of $\autB_{c_1},\dots,\autB_{c_\ell}$ agree on the state space of
  $\autA$, the states of $\autC$ have the form
  $(q,r_1,\dots,r_\ell)$, where $q$ is a state of $\autA$ and the
  $r_i$s are ranks of the $\autB_i$s.
  Thus, $\autC$ has $n(2nw+1)^{\ell}$ states.  An upper bound on
  $\size{\autC}$ is $\bigOh\big(\size{\autA}(nw+1) +
  n(2nw+1)^{\ell}\big)
  \subseteq\bigOh(\size{\autA}(2nw+1)^{\ell}\big)$.  With
  Theorem~\ref{thm:ranking} we conclude that $\autC$ accepts the
  language $L_{\BDAG_{\leq w}}(\autA)$. It is standard to obtain from
  $\autC$ an equivalent B\"uchi automaton $\autB$ with $\bigOh\big(kn
  (2nw+1)^{\lceil k/2\rceil}\big)$ states and
  $\size{\autB}\in\bigOh\big( k\size{\autA}(2nw+1)^{\lceil
    k/2\rceil}\big)$.
\end{proof}

%%% Local Variables:
%%% mode: latex
%%% TeX-master: "main-eptcs-proc"
%%% End:

%% file: aflmuclasses-proc.tex
\section{Collapse Results}
\label{sec:bounded_connectivity}

By consecutively applying the previously presented translations to a
\mucalc sentence, we obtain that \mucalc's alternation hierarchy over
any class only containing BDAGs of width at most $w$ collapses to its
alternation-free fragment, for a fixed $w\in\Nat$.
\begin{theorem}
  \label{thm:aflmuforwABA}
  Let $w\geq 2$ and $\bfU\subseteq\BDAG_{\leq w}$.  For every sentence
  $\phi$, there is an alternation-free sentence $\psi$ of size
  $w^{\bigOh(|\phi|\cdot \ad{\phi})}$ such that
  $L_{\bfU}(\psi)=L_{\bfU}(\phi)$. If $\phi$ is guarded, then the size
  of $\psi$ is $\big(|\phi|\cdot w\big)^{\bigOh(\ad{\phi})}$.
\end{theorem}
\begin{proof}
  Suppose $\phi$ is guarded and let $n:=|\phi|$ and $k:=\ad{\phi}$.
  We construct the parity automaton $\autA_{\phi}$ with
  $\size{\autA_\phi}\in\bigOh(n)$ and $\ind{\autA_\phi}=k+1$
  (Theorem~\ref{thm:mucalc_APA}).  Then, we construct from
  $\autA_{\phi}$ the B\"uchi automaton $\autB_{\phi}$ with
  $\size{\autB_{\phi}}\in(nw)^{\bigOh(k)}$
  (Lemma~\ref{lem:PA_BA}). 
  From $\autB_{\phi}$, we obtain the weak automaton $\autC_{\phi}$
  with $\size{\autC_{\phi}}\in (nw)^{\bigOh(k)}\cdot
  \big(2(nw)^{\bigOh(k)}w+1\big)\subseteq (nw)^{\bigOh(k)}$
  (Lemma~\ref{lem:BA_WA}). Finally, we construct the alternation-free
  sentence $\psi$ with $|\psi|\in(nw)^{\bigOh(k)}$
  (Theorem~\ref{thm:weakABA_altfreemucalc}).  By construction,
  $L_{\BDAG_{\leq w}}(\phi)=L_{\BDAG_{\leq w}}(\psi)$. Since
  $\bfU\subseteq \BDAG_{\leq w}$, we have that
  $L_{\bfU}(\phi)=L_{\bfU}(\psi)$.
  When $\phi$ is not guarded we first transform it into guarded form
  (Lemma~\ref{lem:guarded}), which results in an exponential blow-up.
\end{proof}

In the following, we derive from Theorem~\ref{thm:aflmuforwABA}
further classes of structures over which the alternation hierarchy of
\mucalc collapses to the alternation-free fragment.

\paragraph{Infinite Nested Words} 

Nested words~\cite{Alur2009} extend words with a hierarchical
structure.  Infinite words and infinite nested words when represented
as graphs are BDAGs of width~$1$. We omit the details; instead, see
Figure~\ref{fig:nestedwords} for illustrations, where the set of
actions is $\set{\mathrm{+1}}$ and $\set{\mathrm{+1},\mathrm{jump}}$,
respectively.
\begin{figure}[t]
  \vspace{-.3cm}
  \begin{center}
    \includegraphics[scale=.6]{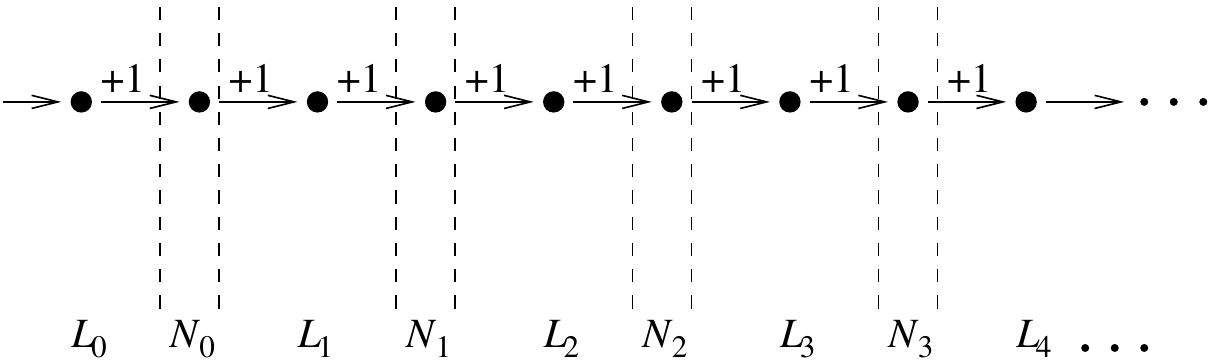}
    \quad
    \includegraphics[scale=.6]{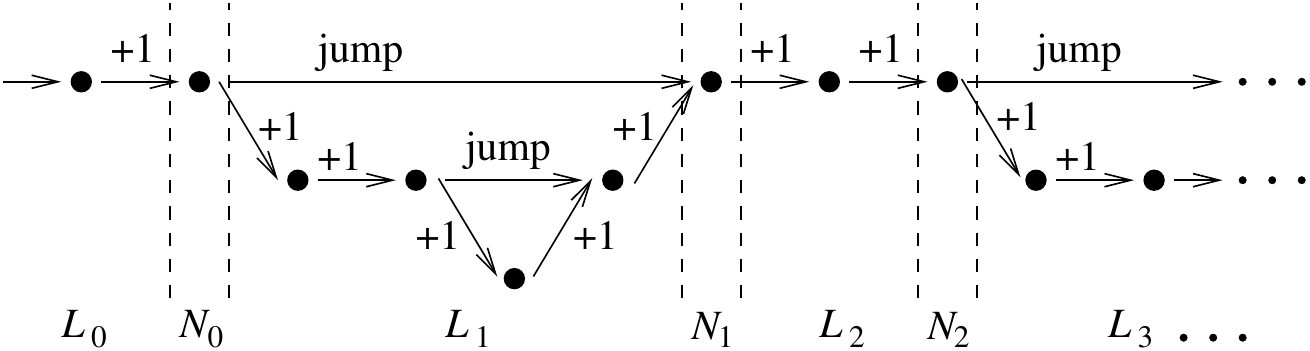}
  \end{center}
  \vspace{-.3cm}
  \caption{\small \label{fig:nestedwords}BDAG representation of
    infinite words (left) and infinite nested words (right)}
  \vspace{-.2cm}
\end{figure}

Then, by Theorem~\ref{thm:aflmuforwABA}, the \mucalc
alternation hierarchy over these structures collapses to the
alternation-free fragment. This improves prior results
in
\cite{Arenas_Barcelo_Libkin2011:nestedwords,Bozzelli2007:nestedwords}
on the expressivity of \mucalc over infinite nested words.

\paragraph{Graphs with Bounded Feedback Sets}

In the following, we consider classes of finite graphs that can be
unfolded to bisimilar BDAGs with bounded width.  The width of these
BDAGs is characterized by a minimal feedback vertex set of the
original folded graph.  
A set $F\subseteq V$ is a \emph{feedback vertex set} (FVS) of $\scrG =
\big(V,(E_a)_{a\in A},v_{\rmI},\lambda\big)$ iff the removal of the
vertices in $F$ separates $\scrG$ into a set of finite DAGs.  Finite
DAGs have the empty set as a feedback vertex set.  We say that a
finite graph $\scrG$ is \emph{$k$-DAG-decomposable} iff the minimal
cardinality of a FVS of $\scrG$ is $k \in \Nat$.
Recall that the $(\Sigma,A)$-graphs $\scrG=\big(V,(E_a)_{a\in
  A},v_{\rmI},\lambda\big)$ and $\scrG'=\big(V',(E'_a)_{a\in
  A},v_{\rmI}',\lambda'\big)$ are \emph{bisimilar} iff there is an
equivalence relation $R\subseteq V\times V'$ with the following
properties: (i)~$\lambda(v)=\lambda'(v')$, for all $(v,v')\in R$,
(ii)~$(v_{\rmI},v'_{\rmI})\in R$, (iii)~for all $u,v\in V$, $u'\in
V'$, and $a\in A$, if $(u,v)\in E_a$ and $(u,u')\in R$ then
$(u',v')\in E'_a$ and $(v,v')\in R$, for some $v'\in V'$, and (iv)~for
all $u',v'\in V'$, $u\in V$, and $a\in A$, if $(u',v')\in E'_a$ and
$(u,u')\in R$ then $(u,v)\in E_a$ and $(v,v')\in R$, for some $v\in
V$.

\newcommand{\unf}{\ensuremath{\mathtt{unf}}\xspace}
\begin{lemma}
  \label{lem:constbdag}
  For every $k$-DAG-decomposable $(\Sigma,A)$-graph $\scrG$, with
  $k\in\Nat$, there is a bisimilar BDAG $\scrD$ of width $k$.
\end{lemma}
\begin{proof}
  Let $\scrG$ be a $(\Sigma,A)$-graph $\big(V,(E_a)_{a\in
    A},v_{\rmI},\lambda\big)$ with a minimal FVS $F\subseteq V$ of
  cardinality $k$.  Furthermore, let $\scrT$ be the tree unfolding of
  $\scrG$ and let $\unf$ be the relation between the vertices of
  $\scrG$ and $\scrT$ that witnesses that $\scrG$ and $\scrT$ are
  bisimilar.  The construction of a bisimilar BDAG $\scrD$ of width
  $k$ is as follows, which is done in a layer-wise manner.

  Construct $\beta$, a (partial and surjective) function between the
  vertices of $\scrT$ and $\scrD$ (i.e., from the elements of the tree
  $\scrT$ to the elements of the acyclic graph $\scrD$), as follows:
  Assume that for each layer $i$, there is a set $S_i$ of \emph{seeds}
  for such a layer. Using $\unf$ collect all vertices in $\scrT$ that
  can be reached from the vertices in $S_i$ until, in every branch,
  (i)~a vertex with no successors is reached or (ii)~two occurrences
  in the unfolding $\scrT$ of a vertex $v\in F\subseteq V$ in $\scrG$
  are found.  Such vertices belong to layer $i$ in $\scrT$. Then,
  $\beta$ maps such a subset of vertices in $\scrT$, denoted by $R_i$,
  to vertices in layer $i$ of $\scrD$ as follows---and let ${\rm
    max}(R_i)$ be the set of maximal or terminal elements in the
  forest $R_i$:
\begin{equation}\label{betamrg}
  \forall u',u''\in {\rm max}(R_i).\ \forall v\in F.\ 
  (v,u')\in\unf \text{ and } (v,u'')\in\unf \Longrightarrow
  \beta(u') = \beta(u'')
  \,.
\end{equation}
Edges in layer $i$ of $\scrD$ are edges in layer $i$ of $\scrT$ which
respect $\beta$, i.e., if $(w',w'')$ in layer $i$ of $\scrT$ via
action~$a\in A$ then $(\beta(w'),\beta(w''))$ in layer $i$ of $\scrD$
via action $a$.  As in $\scrD$ every vertex belongs to $L_i$ or $N_i$,
then the following holds:
\begin{equation}\label{betaNi}
  \forall u\in {\rm max}(R_i).\ \forall v\in F .\ 
  (v,u)\in\unf \Longrightarrow \beta(u)\in N_i
  \,,
\end{equation}
otherwise $\beta(u)\in L_i$. 
In order to define the set of seeds $S_{i+1}$ for the $(i+1)$th layer of $\scrT$, 
firstly one needs to define a subset $F'_i$ of ${\rm max}(R_i)$ whose successors in $\scrT$ 
will be the seeds $S_{i+1}$ of the layer $i$ of $\scrT$. 
The set $F'_i$ is a subset of ${\rm max}(R_i)$ that satisfies two conditions:
\begin{enumerate}
\item[(a)] if $\exists (v,u)\in\unf \text{ with } v\in F \text{ and }
  u\in {\rm max}(R_i)$ then $\exists ! u'\in F'_i \text{ such that } (v,u')\in\unf$, and
\item[(b)] if $u'\in F'_i$ then $\exists (v,u')\in\unf \text{ such
    that } v\in F
  \text{ and } u'\in {\rm max}(R_i)$.
\end{enumerate}
Condition~(a) ensures that 
$F'_i$ contains {\em no more than one} vertex in the unfolding under \unf of a vertex in $F$
as well as that there is one vertex in $F'_i$ for each vertex in ${\rm max}(R_i)$ which 
is associated under \unf with a vertex in $F$.  
Condition (b)~ensures that 
every vertex in $F'_i$ is the occurrence in the unfolding under \unf of a vertex in $F$. 
It is because of condition (a) that the function $\beta$ is partial rather than total.

Edges between vertices in consecutive layers are defined as expected:
if $(u,s)$ in $\scrT$ via action $a\in A$, with $u\in {\rm max}(R_i)$
and $s\in S_{i+1}$, then $(\beta(u),\beta(s))$ in $\scrD$ via action
$a$. The labeling function in $\scrD$ is as in $\scrT$ (and obviously
as in $\scrG$): for any vertex $\beta(u)$ in $\scrD$,
$\lambda_{\scrT}(u) = \lambda_{\scrD}(\beta(u))$.  Finally, $\scrD$ is
constructed recursively using \unf and $\beta$ by letting the set
$S_0$ be the singleton set that only contains the root of $\scrT$.

Clearly, $\scrD$ is a DAG. Within as well as between layers $\beta$
always respects the acyclic structure produced by $\unf$, even when
different occurrences of vertices in $F\subseteq V$ are unified as
they are always terminal elements of a given $R_i$ and thus edges
in $N_i \times L_{i+1}$, i.e.\ the source of edges to the next
layer.
    
Finally, $\scrG$ and $\scrD$ are bisimilar because either (i) a vertex
in $\scrD$ is obtained by a tree unfolding, and every graph is
bisimilar to its own tree unfolding or (ii) a vertex in $\scrD$ is
obtained by unifying occurrences of the same vertex in $\scrG$, and of
course every vertex of a graph is bisimilar to itself. Then,
bisimilarity is preserved when constructing $\scrD$.  To see that
$\scrD$ is a BDAG of width at most $k$ observe that
$\sup\Setx{|N_i|}{i\in\Nat}\leq k$ because every $N_i$ defined by rule
(\ref{betaNi}) cannot contain more than one occurrence of a vertex in
$F$ due to rule (\ref{betamrg}). Then, in fact, there must exist some
$i\in\Nat$ such that $|F|=|N_i|$ since $F$ is minimal. 
\end{proof}

Then, we obtain the following result.

\begin{theorem}
  \label{thm:aflmuforkgraphs}
  Let $k\in\Nat$ and $\bfU$ be a class of $k$-DAG-decomposable
  $(2^{\propset},A)$-graphs. For every sentence $\phi$, there is an
  alternation-free sentence $\psi$ such that
  $L_{\bfU}(\psi)=L_{\bfU}(\phi)$.
\end{theorem}
\begin{proof}
  It follows from Lemma~\ref{lem:constbdag} that each graph in $\bfU$
  can be unfolded into a BDAG of width $k$. We then apply Theorem
  \ref{thm:aflmuforwABA} to this obtained class of unfolded BDAGs.
\end{proof}

Since collapse results carry over to smaller classes of structures,
Theorem~\ref{thm:aflmuforkgraphs} implies the collapse of the
alternation hierarchy over the smaller class of undirected
$k$-DAG-decomposable graphs. 

Finally, we consider classes of graphs that can be decomposed by
removing a bounded number of edges.  Let $\scrG = \big(V,(E_a)_{a\in
  A},v_{\rmI},\lambda\big)$ be a $(\Sigma,A)$-graph.  A set
$F\subseteq \bigcup_{a\in A}E_a$ is a \emph{feedback edge set} (FES)
of $\scrG$ iff the removal of the edges in $F$ separates $\scrG$ into
a set of finite DAGs.
Since every graph with FES $F$ has also a FVS of cardinality at most
$|F|$, we obtain the following corollary.
\begin{corollary}
  Let $k\in\Nat$ and $\bfU$ be a class of finite
  $(2^{\propset},A)$-graphs with minimal FESs of size $k$. For every
  sentence $\phi$, there is an alternation-free sentence $\psi$ such
  that $L_{\bfU}(\psi)=L_{\bfU}(\phi)$.
\end{corollary}

%%% Local Variables:
%%% mode: latex
%%% TeX-master: "main-eptcs-proc"
%%% End:

%% file: concl-proc.tex
\section{Conclusion and Future Work}
\label{sec:conclusion}

The results in this paper focus on \mucalc's expressivity.  By
generalizing and utilizing automata-theoretic methods, we have
unified, generalized, and strengthened prior collapse results of
\mucalc's alternation hierarchy, namely, the results on finite acyclic
directed graphs~\cite{Mateescu:2002:LMC}, infinite
words~\cite{Kaivola1995:multl}, and infinite nested
words~\cite{Arenas_Barcelo_Libkin2011:nestedwords}.
Future work includes to investigate whether our automata construction
for eliminating odd colors in parity automata can be generalized and
to explore over which other classes of structures such generalizations
apply.
The ultimate goal is to characterize the classes of graphs over which
the alternation-free fragment has already the same expressivity as the
full $\mu$-calculus.

We mainly ignore complexity issues in this paper, except the
established upper bounds on the sizes of the resulting
alternation-free formulas.  It remains as future work to provide lower
bounds and to investigate the computational complexity of the
satisfiability problem for \mucalc with respect to classes of
structures over which its alternation hierarchy collapses.

\vspace{-.2cm}
%\enlargethispage{\baselineskip}
\paragraph{Acknowledgments}

The authors thank Christian Dax for initial discussions on the topic
of this paper and Julian Bradfield for advice on the alternation
hierarchy.
% %%
Julian Gutierrez acknowledges with gratitude the support of EPSRC
grant `Solving Parity Games and Mu-Calculi' and ERC Advanced grant
ECSYM.

%%% Local Variables:
%%% mode: latex
%%% TeX-master: "main-eptcs-proc"
%%% End: